\documentclass{article}
\usepackage{spconf}

\usepackage{cite}
\usepackage{amsmath,amssymb,amsfonts}
\usepackage[pdftex]{graphicx}
\graphicspath{{pdf/}}
\DeclareGraphicsExtensions{.pdf}
\usepackage{amsthm}
\newtheorem{theorem}{Theorem}
\newtheorem*{df}{Definition}
\newtheorem{lemma}{Lemma}

\newtheorem{example}{Example}

\usepackage[caption=false,font=footnotesize]{subfig}
\usepackage{url}

\usepackage{xspace}
\newcommand{\sv}{\vect{s}}
\newcommand{\samp}[1][d,\tau]{\ensuremath{\sv_{#1}}\xspace}

\newcommand{\Tr}[1][\tau]{\ensuremath{\mathcal{T}_{#1}}\xspace}
\newcommand{\Qc}{\ensuremath{\mathcal{Q}}\xspace}
\newcommand{\Rt}[1][\tau]{\ensuremath{\mathcal{R}_{#1}}\xspace}
\newcommand{\R}{\ensuremath{\mathbb{R}}\xspace}

\newcommand{\ccurve}[1][d,\tau]{\ensuremath{\mathcal S_{#1}}\xspace}

\newcommand{\vect}[1]{\ensuremath{\mathbf{#1}}\xspace}

\newcommand{\pv}{\vect{p}}

\newcommand{\qv}{\vect{q}}

\newcommand{\lpr}[1]{\ensuremath{\underleftarrow{#1}}\xspace}
\newcommand{\rpr}[1]{\ensuremath{\underrightarrow{#1}}\xspace}

\begin{document}
%
\title{Period and Signal Reconstruction\\ from the Curve of Trains of Samples}
%
%
%

\name{Marek W. Rupniewski}
\address{Institute of Electronic Systems\\
	Warsaw University of Technology,\\
	Nowowiejska 15/19, 00-665 Warsaw, Poland\\
  Email: Marek.Rupniewski@pw.edu.pl%
}%

\maketitle

\begin{abstract}
    A finite sequence of equidistant samples (a~sample train) of a periodic signal
    can be identified with a~point in a~multi-dimensional space. Such a~point
    depends on the sampled signal, the sampling period, and the starting time
    of the sequence. If the starting time varies, then the corresponding point
    moves along a closed curve. We prove that such a curve, i.e., the set of all
    sample trains of a given length, determines the period of the sampled
    signal, provided that the sampling period is known. This is true even if the
    trains are short, and if the samples comprising trains are taken at
    a~sub-Nyquist rate. The presented result is proved with a~help of the theory of
    rotation numbers developed by Poincaré.
    We also prove that the curve of sample trains determines the sampled signal up
    to a~time shift, provided that the ratio of the sampling period to the
    period of the signal is irrational. Eventually, we give an example which
    shows that the assumption on incommensurability of the periods cannot be dropped.
\end{abstract}


\section{Introduction}
A short train of samples taken at sub-Nyquist frequency rarely allows for
determining the period of a continuous signal. Clearly, it is not
possible unless some sparsity conditions are imposed on the class of the
considered signals. Intuitively, any extra sample train should help in
finding a better estimate for the unknown period even if the time offsets
between the trains are not known. The aim of this paper is to augment the
intuition by showing
that the period of the~signal is uniquely determined by the set of all
sample trains of a~given length, provided that the sampling period is known.
For this, we embed the sample trains into a~multi-dimensional space:
\begin{equation}\label{e:seqdef}
   \samp[d,\tau](t) = [s(t),\, s(t+\tau),\, \dots,\, s(t+(d-1)\tau)]\in\R^d,
\end{equation}
where $s$ is the signal being sampled, $d$ is the length of the train, $\tau$ is the inter-sample distance, i.e., the sampling period, and $t$ is the starting time of the train.
Such an embedding was proposed by Packard et al.\ \cite{Packard80} and by Takens~\cite{Takens81} in order
to study the dynamics of nonlinear systems. It has also been used in
topological data analysis~proposed by Carlsson~\cite{Carlsson09} to identify
qualitative properties of the data-sets. We denote the set of all points
\eqref{e:seqdef} by \ccurve, i.e.,
\begin{equation}
    \ccurve = \bigl\{ \samp(t) \colon  t\in\R \bigr\} \subset \R^d.
\end{equation}
If signal $s$ is continuous and periodic, then \ccurve is a closed curve.
We show that if a $T$-periodic signal $s$ satisfies some
regularity conditions, then curve $\ccurve\subset\R^d$ defines
a~homeomorphism \Rt of curve $\ccurve[d-1,\tau]\subset\R^{d-1}$ which is formed by sample trains of
length~$d-1$. Poincaré \cite{poincare81} showed that each homeomorphism of
a~closed curve has a~topological invariant called the rotation number. In this
paper, we prove that the rotation number of the homeomorphism \Rt is strictly related
to the ratio~$\tau/T$. This relationship can be used to find
signal period~$T$ if curve~\ccurve and sampling period $\tau$ are known. 

Works of Rader \cite{rader77} and Choi et al.\ \cite{choi11} showed that
$T$-periodic signals can be reconstructed from infinite trains of their samples
taken with any sampling period $\tau<\frac{T}{2}$ such that $\tau/T$ is irrational. 
In this paper, we prove that a~similar reconstruction is possible from the
infinite set of finite-length sample trains, i.e., we may reconstruct
a~periodic signal from its curve~\ccurve. In~\cite{Rup20icassp} and
\cite{Rup20spl} the author demonstrated that periodic signals can be reconstructed
from a~probabilistic distribution defined on their curves~\ccurve. The results
presented in this paper imply that such a~distribution is not needed if the
sampling period and the signal period are incommensurable.

The next two sections introduce the concepts of a periodic map and a~rotation
number, respectively. Section~\ref{s:perrec} shows how to reconstruct the
period of the signal from the curve of sample trains, and the following section
deals with the reconstruction of signals up to a~time shift. The paper is
concluded in Section~\ref{s:conclusion}.

\section{Periodic covering maps}\label{s:covmap}
In this and the following sections we make use of some concepts of algebraic topology and
dynamical systems, e.g.,
a~homeomorphism, a~covering map, a~lift and a~rotation number. Because of the
scope and the form of this paper, we reduce the presentation of the required mathematical concepts to a~bare minimum. A more comprehensive treatment of these concepts can be found
in, e.g., \cite{Spanier1989, demelo2012}. 

\begin{df}
    A $T$-periodic mapping $\sv\colon\R\to\R^d$ is called a~covering map if  
    $\sv$ restricted to any interval of length smaller than $T$ is a~homeomorphism, i.e.,
    if each such restriction is continuous and has continuous inverse function.
\end{df}
If a $T$-periodic mapping $\sv\colon\R\to\R^d$ is a covering map,
then the image of $\sv$ does not have self-intersections and thus this image is
a~closed curve which is homeomorphic to a~circle. The covering map carries the natural
orientation of real line (from smaller to bigger numbers) to one of the two
possible orientations of this closed curve, see~Fig.~\ref{f:rys02}. 

\begin{figure}[tpb]
    \centerline{\includegraphics{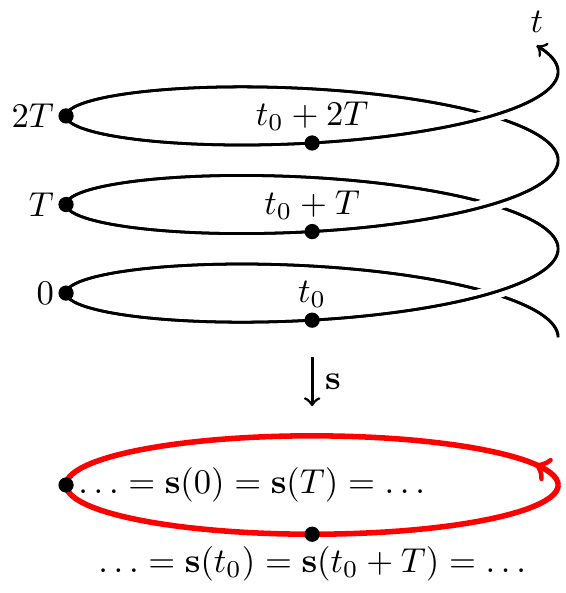}}
    \caption{The orientation of a closed curve induced by its $T$-periodic covering map.}
    \label{f:rys02}
\end{figure}

The following example shows that whether
or not a mapping \samp is a~covering map may depend on the sample train length~$d$.
\begin{example}\label{ex:1}
    Let $s\colon\R\to\R$ be a $1$-periodic signal defined as
    \begin{equation}\label{e:ex1}
        s(t) = (\sin\pi t)^2 \sin\left(2\pi t \left(1+t\right)\right),\qquad \text{for }t\in[0,1].
    \end{equation}
    Figure~\ref{f:rys01a} shows the graph of signal~$s$.
    Mapping \samp[2,0.2] is not a~covering map because curve \ccurve[2,0.2] has intersections as shown in
    Fig.~\ref{f:rys01b}.
    However,
    \samp[d,0.2] becomes a~covering map for $d=3$ (see Fig.~\ref{f:rys01c}), and it stays
    a~covering map for all $d>3$.
\end{example}
\begin{figure}[tpb]
    \centerline{\includegraphics{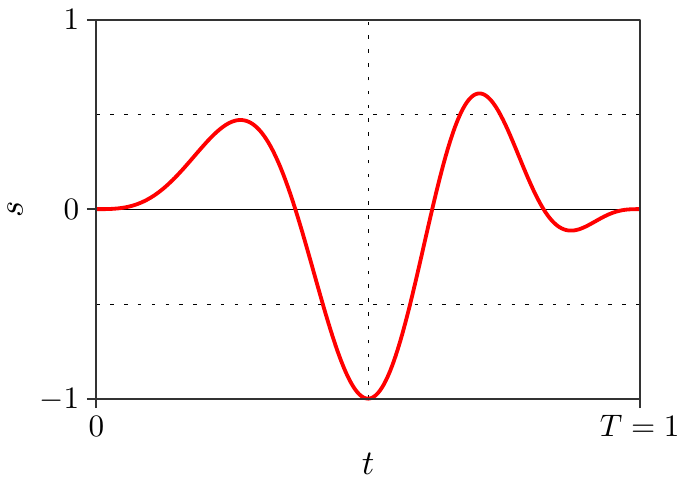}}
    \caption{A $1$-periodic signal $s$ considered in Example~\ref{ex:1}.}
    \label{f:rys01a}
\end{figure}
The last statement of the above example results from the following lemma, which
is a direct consequence of~\eqref{e:seqdef} and the definition of a~periodic
covering map.
\begin{lemma}\label{l:scovind}
If $s$ is a signal for which \samp is a $T$-periodic covering map, then
\samp[d+1,\tau] is also a~$T$-periodic covering map.
\end{lemma}
\begin{figure*}[th]
    \centering%
    \subfloat[$d=2$]{\includegraphics{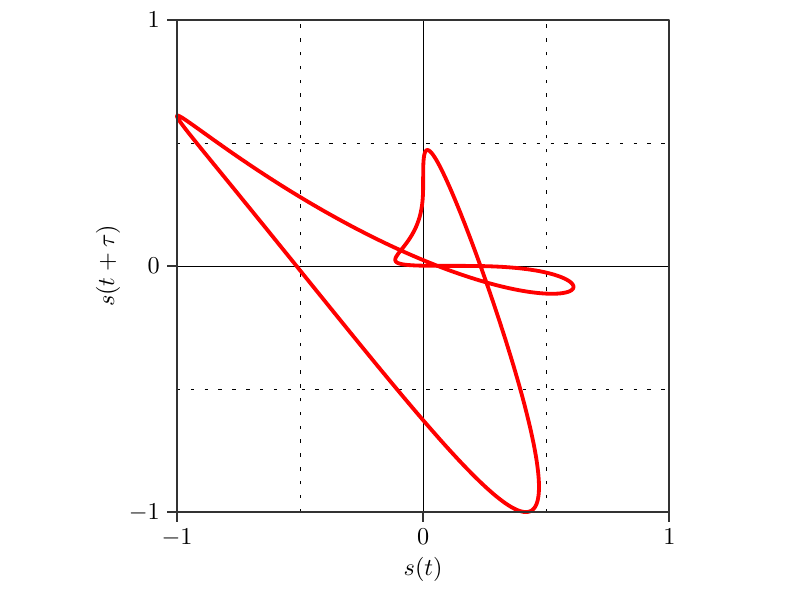}\label{f:rys01b}}%
    \subfloat[$d=3$]{\includegraphics[width=11cm]{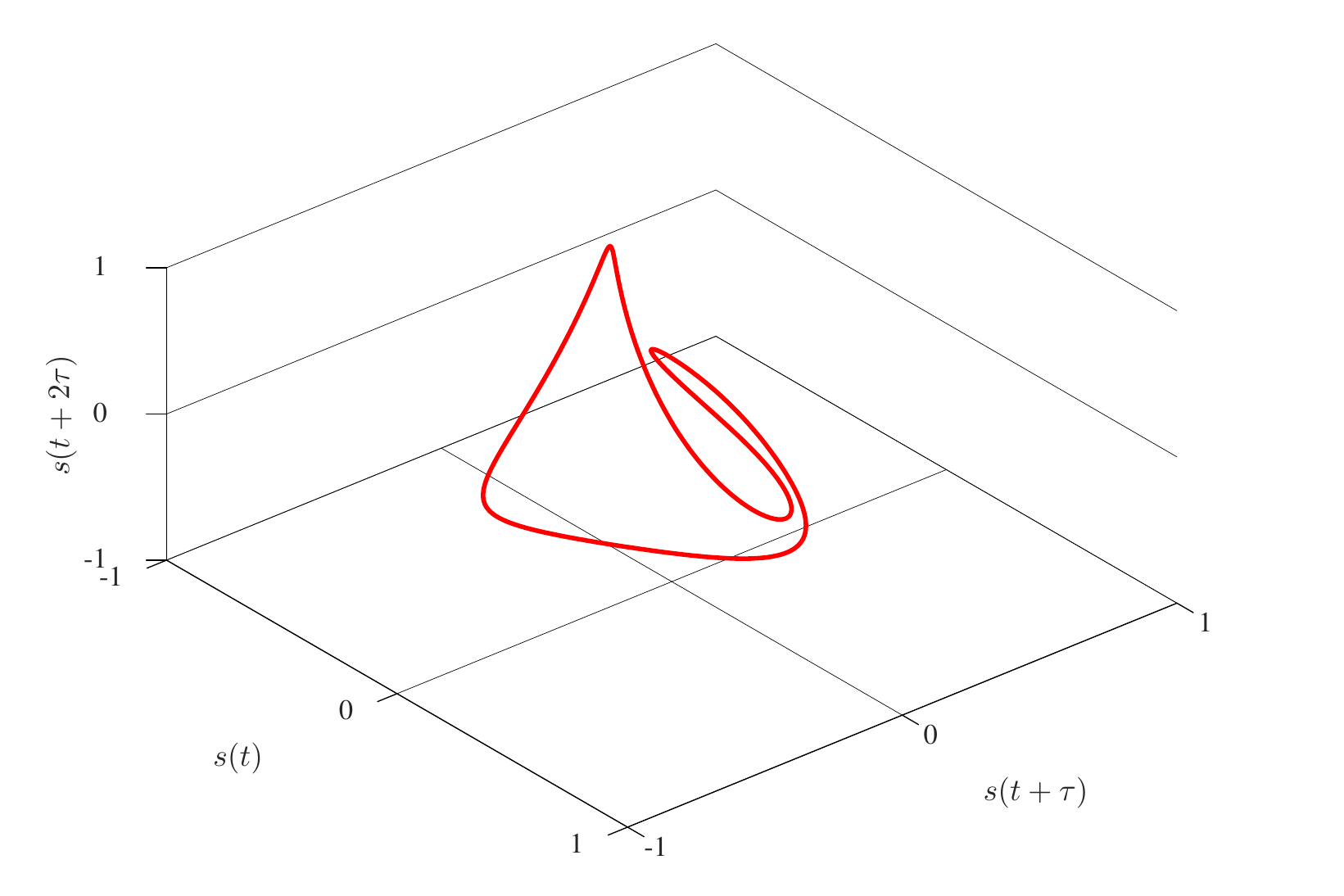}\label{f:rys01c}}%
    \caption{Curves \ccurve[d,0.2] for signal $s$ considered in Example~\ref{ex:1}.}
    \label{f:rys01bc}
\end{figure*}


\section{Rotation number}\label{s:rotnum}
Let a~curve \ccurve[] be the image of a $T$-periodic covering map $\sv\colon\R\to\R^d$
and let $\Rt[]\colon\ccurve[]\to\ccurve[]$ be a~homeomorphism.
\begin{df}
    A mapping $F\colon\R\to\R$ is called a~lift of \Rt[] across~\sv if
    \begin{equation}
        \Rt[] \circ \sv = \sv \circ F,
    \end{equation}
    where symbol $\circ$ denotes function composition.
\end{df}
If \Rt[] preserves the orientation of \ccurve[] induced by \sv, then the lift of \Rt[] across~\sv is an increasing function and 
\begin{equation}
    F(x+kT) = F(x) + kT,
\end{equation}
for every real $x$ and integer $k$.
Similarly, if \Rt[] reverses the orientation of \ccurve[], then $F$ is decreasing and $F(x+kT)=F(x)-kT$.

If \samp is a $T$-periodic covering map, then there exists a~natural homeomorphism $\Rt$ of curve \ccurve:
\begin{equation}\label{e:Rtdef0}
    \Rt\colon\ccurve\to\ccurve\qquad \Rt\left(\samp\left(t\right)\right) = \samp\left(t+\tau\right).
\end{equation}
Translation of real line by $\tau$:
\begin{equation}
    \Tr\colon\R\to\R,\qquad \Tr(x)=x+\tau
\end{equation}
is a~lift of \Rt across \samp, i.e.,
\begin{equation}\label{e:Rtdef}
    \Rt \circ \samp = \samp \circ \Tr.
\end{equation}

\begin{df}
    Rotation number of an orientation-preserving homeomorphism \Rt[] of an oriented~closed
    curve is defined as the limit 
    \begin{equation}\label{e:rhodef}
        \rho(\Rt[]) = \lim_{n\to\infty} \frac{F^n(x)-x}{nT} \mod 1,
    \end{equation}
    where $x\in\R$ is an arbitrary point, $F\colon\R\to\R$ is a lift of~\Rt[]
    across a~$T$-periodic covering map~$\sv$ which induces the orientation of the curve,
    and where $F^n$ stands for $n$-fold composition of $F$ with itself.
\end{df}
Poincaré \cite{poincare81} showed that 
the choice of point $x$, lift $F$ and covering map~\sv does not affect the value
of~\eqref{e:rhodef} (see also~\cite{demelo2012} for a~comprehensive treatment
of the subject). He also showed that the rotation number is invariant under
topological conjugacy, i.e., if \Rt[] and \Qc are two homeomorphisms of the same
closed curve, then the rotation number of \Rt[] is the same as the rotation
number of $\Qc^{-1} \circ \Rt[] \circ \Qc$.
Note that the rotation number depends on the orientation of the curve, i.e., for a~given
homeomorphism \Rt[] the rotation number changes from $\rho$ to $1-\rho$ if
the considered orientation of the curve is reversed.

Substitution $\sv=\samp$, $F=\Tr$ and $\Rt[]=\Rt$ into~\eqref{e:rhodef} yields
\begin{equation}
    \rho(\Rt) = \frac{\tau}{T} \mod 1,
\end{equation}
provided that the chosen orientation of curve \ccurve happens to be the same as the orientation induced by~\samp. In the opposite case, we get
\begin{equation}
    1-\rho(\Rt) = \frac{\tau}{T} \mod 1.
\end{equation}

We conclude the above consideration with the following lemma. 
\begin{lemma}\label{l:perdet}
    If $s$ is a~$T$-periodic signal such that \samp is a~covering map and $\tau<\frac{T}{2}$,
    then 
    \begin{equation}
        \frac{\tau}{T} = \min\bigl(\rho(\Rt), 1-\rho(\Rt)\bigr),
    \end{equation}
    where \Rt and $\rho(\Rt)$ are defined
    by~\eqref{e:Rtdef} and~\eqref{e:rhodef}, respectively.
\end{lemma}

\section{Period reconstruction}\label{s:perrec}
Let $s$ be a periodic signal. If its curve \ccurve[d+1,\tau] is given, then we
can recover curve \ccurve by projecting \ccurve[d+1,\tau] onto the first $d$ or the last
$d$ coordinates, i.e.,
\begin{equation}
    \ccurve = \{\lpr{\pv}\mid \pv\in\ccurve[d+1,\tau]\}
            = \{\rpr{\pv}\mid \pv\in\ccurve[d+1,\tau]\},
\end{equation}
where for any point $\pv=[x_1,\dots,x_{d+1}]\in\R^{d+1}$, the projections are defined as
\begin{equation}
    \lpr{\pv} = [x_1,\dots,x_d]\quad\text{and}\quad
    \rpr{\pv} = [x_2,\dots,x_{d+1}].
\end{equation}
If \samp happens to be a~covering map, then homeomorphism
$\Rt\colon\ccurve\to\ccurve$ is determined by curve \ccurve[d+1,\tau]
(with some abuse of notation we denote by the same symbol \Rt all the homeomorphisms defined by~\eqref{e:Rtdef0} independently of the value of~$d$).
The above fact is implied by Lemma~\ref{l:scovind} and the observation stated below as
Lemma~\ref{l:Rt}.
\begin{lemma}\label{l:Rt}
    If $s$ is a periodic signal such that \samp is a~covering map, then
    the homeomorphism $\Rt\colon\ccurve\to\ccurve$ defined by~\eqref{e:Rtdef0}, 
    is uniquely determined by \ccurve[d+1,\tau]:
    \begin{equation}\label{e:Rtdefgeom}
        \Rt (\lpr{\pv}) = \rpr{\pv},\quad\text{for each}\quad \pv \in\ccurve[d+1,\tau].
    \end{equation}
\end{lemma}

Lemmas~\ref{l:perdet} and~\ref{l:Rt} imply the following theorem, which provides
conditions under which curve \ccurve[d+1,\tau] determines the period of
a~periodic signal~$s$.
\begin{theorem}\label{t:T}
    Let $s$ be a $T$-periodic signal such that \samp is a~covering map.
    If $\tau<\frac{T}{2}$, then curve \ccurve[d+1,\tau] determines period $T$ as
    \begin{equation}
        T = \frac{\tau}{\min\bigl(\rho\left(\Rt\right), 1-\rho\left(\Rt\right)\bigr)},
    \end{equation}
    where homeomorphism \Rt is defined by~\eqref{e:Rtdefgeom}.
\end{theorem}

Notice that without condition $\tau<\frac{T}{2}$ in the assertion of the above theorem,
all that could be deduced from \ccurve is that
\begin{equation}
    T = \frac{\tau}{n + \rho\left(\Rt\right)}\quad\text{or}\quad
    T = \frac{\tau}{n + 1 - \rho\left(\Rt\right)},
\end{equation}
for some non-negative integer~$n$. In other words, we may replace condition $\tau<\frac{T}{2}$ in Theorem~\ref{t:T} with the requirement that 
\begin{equation}\label{e:newcond}
    \tau \in \left( n \frac{T}{2},\, (n+1)\frac{T}{2}\right),
\end{equation}
for any non-negative integer~$n$, provided that this integer is known.

\section{Signal reconstruction}\label{s:sigrec}
In the previous section, we showed that the period of a~periodic signal may be
recovered from its curve \ccurve[d+1,\tau] formed by all sample trains of length $d+1$ which are 
taken with a~fixed sampling period~$\tau$. An interesting question is whether it is
also possible to reconstruct the signal $s$ from that curve. The answer to this
question depends on the rotation number of the homeomorphism~\Rt.
\begin{theorem}\label{t:s}
    Let $s$ be a $T$-periodic signal such that \samp is a~covering map for some
    $\tau<\frac{T}{2}$, and let $\Rt\colon\ccurve\to\ccurve$ be the homeomorphism defined
    by~\eqref{e:Rtdefgeom}. If rotation number $\rho(\Rt)$ is irrational, then
    curve \ccurve[d+1,\tau] determines signal $s$ up to a~time shift.
\end{theorem}
\begin{proof}
    By Theorem~\ref{t:T} we can assume that period~$T$ is known. Let \qv be
    a~point lying in~\ccurve. This point equals $\samp(t_0)$ for some
    $t_0\in[0,T)$. We set
    \begin{equation}
        s_1(0) = \pi_1(\qv) = s(t_0),
    \end{equation}
    where $\pi_1$ denotes the projection to the first coordinate, i.e.,
    $\pi_1([x_1,\dots,x_d])=x_1$.
    Without loss of generality, we may assume that the orientation of \ccurve is so that~$\rho(\Rt)=\tau/T$.
    Since the rotation number $\rho(\Rt)$ is irrational, multiples
    of~$\tau$ form a~dense subset of interval $[0,T)$ when considered modulo~$T$~\cite{demelo2012}.
    Therefore, for any $t\in\R$ there exists a sequence $k_1,\,k_2,\,\dots$ of
    positive integers such that
    \begin{equation}
        \lim_{n\to\infty} {k_n}\tau = t \mod T.
    \end{equation}
    We set
    \begin{equation}
        s_1(t) = \pi_1\left(\lim_{n\to\infty} \Rt^{k_n}\left(\qv\right)\right) =
        s(t_0 + t).
    \end{equation}
    Thus, we obtain a~signal $s_1\colon\R\to\R$, which differs from $s$ by a~time
    shift only.
\end{proof} 
If the rotation number of homeomorphism~\Rt is rational, then \ccurve is not
enough to recover signal $s$ up to a~time shift, what is shown in the following
example.
\begin{example}\label{ex:2}
    Let $\tau=p/q$ for some positive integers~$p<q$. Signal $s(t)=\sin 2\pi t$
    is a~$1$-periodic signal and \samp is a~covering map for each $d\geq2$
    (\ccurve[2,\tau] is an~ellipse). Let $r$ be a homeomorphism of~\R defined as
    \begin{equation}
        r(t) = t + \frac{1}{4\pi q}\sin\left(2\pi q t\right).
    \end{equation}
    Signal $s'$ defined as $s'=s\circ r$ is another $1$-periodic signal and $\samp'=\samp \circ r$
    is a covering map for each $d\geq2$. Signals $s$ and $s'$ do
    not differ by a~time shift only (see Fig.~\ref{f:rys03}). However, curves \ccurve defined for
    signal~$s$ are the same as the corresponding curves $\ccurve'$ constructed for $s'$ because
    the image of $\samp'$ is the same as the image of $\samp$ for each
    $d\geq2$.
\begin{figure}[t]
    \centerline{\includegraphics{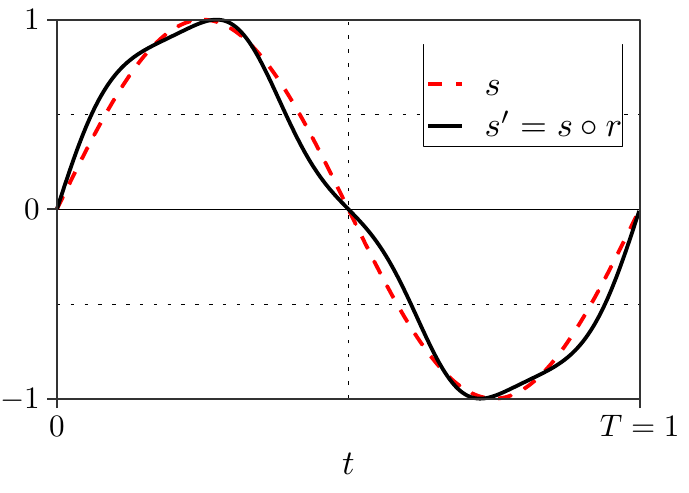}}
    \caption{Signals $s$ and $s'$ considered in Example~\ref{ex:2} for $q=3$.}
    \label{f:rys03}
\end{figure}
\end{example}
It is worth to note that, although \ccurve is not enough to recover the
underlying signal~$s$, it is possible to reconstruct signal $s$ from the
probability distribution on \ccurve which results from uniform
probability distribution of the starting times of the
sample trains~\cite{Rup20icassp,Rup20spl}.

\section{Conclusion}\label{s:conclusion}
We showed that a curve \ccurve formed by sample trains of length~$d\geq3$ carries
valuable information on the sampled signal~$s$. If mapping \samp[d-1,\tau] is a~covering
map, then the period $T$ of the signal~$s$ can be recovered from curve~\ccurve,
provided that $2\tau/T<1$ (Theorem~\ref{t:T}). The same holds if
$2\tau/T\in(k,k+1)$ for any positive integer~$k$, provided that $k$ is known prior to the recovery of~$T$.
Moreover, if ratio $\tau/T$ is irrational then \ccurve determines signal~$s$ up to
a time shift (Theorem~\ref{t:s}). By providing a counterexample we showed that
the same statement does not hold when $\tau/T$ is rational.

The results presented in this paper establish a~link between sampling theory,
 the theory of dynamical systems, and algebraic topology. We believe that the
results will also form a basis for further research
on construction of consistent period estimators from finite number of trains of samples.




\bibliographystyle{IEEEbib}
\bibliography{IEEEabrv,biblio}
\end{document}